\pdfoutput=1
\RequirePackage{ifpdf}
\ifpdf 
\documentclass[pdftex]{sigma}
\else
\documentclass{sigma}
\fi

\def\a{\alpha}
\def\b{\beta}
\def\d{\delta}

\def\g{\gamma}
\def\l{\lambda}

\def\e{\epsilon}

\numberwithin{equation}{section}

\newtheorem{Theorem}{Theorem}[section]
\newtheorem*{Theorem*}{Theorem}

\newtheorem{Proposition}[Theorem]{Proposition}
\newtheorem{Conjecture}[Theorem]{Conjecture}
 { \theoremstyle{definition}

\newtheorem{Remark}[Theorem]{Remark} }

\begin{document}
\allowdisplaybreaks

\newcommand{\arXivNumber}{2508.?????}

\renewcommand{\PaperNumber}{066}

\FirstPageHeading

\ShortArticleName{Symmetric Separation of Variables for the Extended Clebsch and Manakov Models}

\ArticleName{Symmetric Separation of Variables\\ for the Extended Clebsch and Manakov Models}

\Author{Taras SKRYPNYK}

\AuthorNameForHeading{T.~Skrypnyk}

\Address{Bogolyubov Institute for Theoretical Physics, 14-b Metrolohichna Str., Kyiv, 03680, Ukraine}
\Email{\href{mailto:tskrypnyk@bitp.kiev.ua}{tskrypnyk@bitp.kiev.ua}}

\ArticleDates{Received September 09, 2024, in final form July 30, 2025; Published online August 05, 2025}

\Abstract{In the present paper, using a modification of the method of vector fields $Z_i$ of the bi-Hamiltonian theory of separation of variables (SoV), we construct {\it symmetric non-St\"ackel} variable separation for three-dimensional extension of the Clebsch model, which is equivalent (in the bi-Hamiltonian sense) to the system of interacting Manakov (Schottky--Frahm) and Euler tops. For the obtained symmetric SoV (contrary to the previously constructed asymmetric one), all curves of separation are the same and have genus five. It occurred that the difference between the symmetric and asymmetric cases is encoded in the different form of the vector fields~$Z$ used to construct separating polynomial. We explicitly construct coordinates and momenta of separation and Abel-type equations in the considered examples of symmetric SoV for the extended Clebsch and Manakov models.}


\Classification{37J35; 37J37; 37J39}

\section{Introduction}
\subsection{Generalities}
Completely integrable Hamiltonian systems have been the objects of constant interest in theoretical and mathematical physics for more than one hundred years. Nevertheless, many problems in this theory still remain unsolved. One of the most important problems to be solved in general is the problem of variable separation. The separated variables $q_i$, $p_j$ are a set of (quasi)canonical coordinates such that the following system of equations (equations of separation) is satisfied~\cite{SklSep}:\looseness=1
\[
 \Phi_i(q_i, p_i,I_1, \dots ,I_n, C_1, \dots ,C_m) = 0, \qquad i\in \{1,\dots, n\}.
 \]
Here $\Phi_i$ are certain functions, $I_k$ are Poisson-commuting integrals of motion, $C_i$ are Casimir functions
and $n$ is half of the
dimension of the phase space.
The separated coordinates provide a~possibility to solve explicitly the Hamilton equations of motion upon resolving the Abel--Jacobi inversion problem. Separated variables are also important when solving quantum
integrable models \cite{SklSep}. That is why separation of variables (SoV) is a central issue in the theory of classical and quantum integrable systems.

 The most simple and investigated is the so-called St\"ackel-type SoV, when all the equations of separation are linear in the integrals and Casimir functions. In the present paper, we are interested in {\it non-St\"ackel} SoV that naturally arises for the integrable systems associated with the Lie algebras $\mathfrak{so}(2n)$.

There exist three main approaches to the variable separation: a classical one going back to the papers of St\"ackel \cite{Sta1,Sta2}, Levi-Civita \cite{Levi} and Agostinelli \cite{Agost} and developed later in the papers of Benenti and his school \cite{Benenti, Chanu1,Chanu2} and two modern ones. They are: the
 ``magic recipe'' of Sklyanin \cite{SklSep} and the bi-Hamiltonian approach of Magri, Falqui and Pedroni \cite{MFP,FP,FP2}.
 In the present paper, we develop the bi-Hamiltonian approach to SoV theory in its formulation based on the theory of the vector fields $Z_i$ \cite{MFP,FP,FP2}.
The theory of the vector fields $Z_i$ permits to encode the information on separated coordinates $q_i$ for the integrable bi-Hamiltonian systems into a set of differential conditions on the vector fields $Z_i$ that are satisfied by them with respect to two compatible Poisson brackets $\{\ ,\ \}_1$ and $\{\ ,\ \}_2$ \cite{MFP,FP,FP2}.

The bi-Hamiltonian approach to SoV (despite the mathematical beauty and good generality) has a weak point: in order to find the vector fields $Z_i$, one should resolve a complicated system of nonlinear PDEs. In some cases, this difficulty can be overcome. Indeed, in a series of our papers \cite{MagSkr, SkrJGP2019, SkrSto, SkrBD} we have proposed to impose a certain simplifying condition onto one of the vector fields $Z_i$, which transforms part of the PDEs from \cite{MFP,FP,FP2} into algebraic equations.\looseness=1

Another problem is in the fact that
the PDEs from \cite{MFP,FP,FP2} ``algebrized'' in \cite{MagSkr, SkrJGP2019, SkrSto, SkrBD} imply that the variable separation are of {\it St\"ackel type}, i.e., that all equations of separation are {\it linear} in the integrals of motion and Casimir functions. For {\it non-St\"ackel SoV} (when some of the equations of separation are {\it nonlinear} in some of the integrals), the corresponding equalities from \cite{MFP,FP,FP2} do not hold true.
Here the natural question arises: how to proceed with the method of the vector fields $Z_i$ in this case? What differential equation to ``algebrize''?\looseness=1

In our previous paper \cite{SkrAsSoVExCle}, we have answered the above question for a special subcase of general non-St\"ackel SoV.
For this purpose we have assumed certain special form of the equations of separation
and shown
that under this condition there exists an invariance vector field~$Z$ (combination of the vector fields $Z_i$, $i\in \{1,\dots, m\}$)
 that satisfies a differential condition replacing the equations on $Z_i$ obtained in \cite{MFP,FP,FP2} from the St\"ackel-type assumption on the equations of separation.
In order to be able to solve this condition we assume (similar to what was done in the St\"ackel case \cite{MagSkr, SkrJGP2019, SkrSto, SkrBD}) that the vector field $Z$ is a special one: it annihilates all its components in the initial system of coordinates. We call such vector fields ``algebraic''. As a result, we obtain a system of {\it quadratic algebraic equations} which is possible to solve explicitly. Moreover, it occurred that the obtained system of quadratic algebraic equations may have {\it different solutions} leading to {\it non-equivalent} separations of variables for the same integrable bi-Hamiltonian system. In the present paper, we illustrate this interesting phenomenon by the example of the extended Clebsch and Manakov models.\looseness=1

The extended Clebsch and Manakov model is an integrable
 bi-Hamiltonian system on a~nine-dimensional space coinciding with a~three-dimensional extension of $\mathfrak{e}^*(3)$ (with respect to the first Poisson brackets $\{\ ,\ \}_1$) or on a three-dimensional extension of $\mathfrak{so}^*(4)$ (with respect to the Poisson brackets $\{\ ,\ \}_2$). It possesses six Poisson-commuting integrals of motion $I_1=H$, $I_2=K$, $I_3=L$, $C_1$, $C_2$, $C_3$, where $C_1$, $C_2$, $C_3$ are Casimir functions of the bracket~$\{\ ,\ \}_1$.\looseness=1

 The described above modification of the method of vector fields $Z_i$ leads to a system of six quadratic equations on nine components of the vector field $Z$. In the paper \cite{SkrAsSoVExCle}, we have found one (highly non-trivial) solution of these quadratic equations. We have shown that it leads to ``asymmetric'' non-St\"ackel SoV characterised by two different separation curves $\mathcal{C}$ and $\mathcal{K}$ of different genus. The last fact makes the solution of the Abel--Jacobi inversion problem to be very complicated \cite{FedMagSkr2}.

 In this paper, we find another (also highly nontrivial) type of solution of the quadratic equations for the components of the vector field $Z$ and show that it leads to three pairs of (quasi)canonical separated variables $p_i$, $q_i$, $i\in \{1,2,3\}$, that
satisfy {\it the same curve of separation}~$\mathcal{K}$ of genus~five
\begin{gather}
(q_i+j_1)(q_i+j_2)(q_i+j_3) p_i^4+\big(q_i^3 C_3+q_i^2 C_2+ q_i H+K \big) p_i^2+ \frac{1}{4}(q_i C_1+L)^2=0,
\label{eqsepi0}
\end{gather}
It leads, in turn, to the Abel-type equations (quadratures) written in the differential form as follows:
\begin{subequations}\label{AbelCleExt0}
\begin{gather}
\sum\limits_{i=1}^3 \frac{2 q_i p_i^3 \mathrm{d}q_i}{ 4 (q_i+j_1) (q_i+j_2)(q_i+j_3) p^4_i-(q_i C_1+L)^2} = \mathrm{d} t_1,
\\
\sum\limits_{i=1}^3 \frac{2 p_i^3 \mathrm{d} q_i}{ 4 (q_i+j_1) (q_i+j_2)(q_i+j_3) p^4_i-(q_i C_1+L)^2} = \mathrm{d} t_2,
\\
\sum\limits_{i=1}^3 \frac{ (q_i C_1+L) p_i \mathrm{d} q_i}{ 4 (q_i+j_1) (q_i+j_2)(q_i+j_3) p^4_i-(q_i C_1+L)^2} = \mathrm{d} t_3,
\end{gather}
\end{subequations}
where $t_1$, $t_2$, $t_3$ are the parameters along the time flows of the integrals $H$, $K$ and $L$ correspondingly.
The fact that the differentials in \eqref{AbelCleExt0} are defined on the same curve, i.e., that the obtained SoV is ``symmetric'', makes the task of solution of the corresponding Abel--Jacobi inversion problem more plausible.

 The new ``algebraic'' vector field $Z$ written in terms of the initial dynamical variables, generated by it separating polynomial, the explicit formula for the momenta of separation, non-St\"ackel equations of separation \eqref{eqsepi0}, as well as symmetric Abel-type equations \eqref{AbelCleExt0} are the main results of the present article.

 The structure of the present paper is the following. In Section~\ref{sec2}, we present the general theory of SoV and the method of the vector fields $Z_i$ in St\"ackel and non-St\"ackel cases. In Section~\ref{sec3}, we~describe three-dimensional extension of the Clebsch and Manakov models, in Section~\ref{sec4}, we~find the corresponding algebraic vector field $Z$. In Section~\ref{sec5}, we construct symmetric SoV for the considered extensions of the Clebsch and Manakov models. In Section~\ref{sec6}, we conclude and describe open problems.

\section{Separation of variables}\label{sec2}

\subsection{Generalities}\label{general}
Let us recall the definitions of Liouville integrability and
separation of variables in the general theory of Hamiltonian
systems. An integrable Hamiltonian system with $n$ degrees of
freedom is determined on a $2n$-dimensional symplectic manifold
$\mathcal{M}$, embedded as a symplectic leaf in a~Poisson
manifold $(\mathcal{P},\{\ ,\ \}_1)$ as a level surface of $m$
Casimir functions $C_i$, by $n$ independent nontrivial integrals
$I_j$ commuting with respect to the Poisson bracket
$
\{I_i ,I_j\}_1 = 0$, $ i,j \in \{1, \dots, n\}$.
 To find separated variables means to find (at least locally) a set of coordinates~$q_i$,~$p_j$, $i,j \in
\{1, \dots, n\}$ such that there exist $n$ relations
\begin{equation}\label{sepeq0}
 \Phi_i(q_i, p_i,I_1, \dots ,I_n, C_1, \dots ,C_m) = 0, \qquad i\in \{1, \dots, n\},
\end{equation}
and the coordinates $q_i$, $p_j$, $i,j \in \{1, \dots, n\}$ are (quasi)canonical
\[
\{ p_i, q_j \}_1= f_i(q_i,p_i) \d_{ij}, \qquad \{q_i,q_j\}_1=0,\qquad
\{p_i,p_j\}_1=0, \qquad \forall i,j \in \{1, \dots, n\}
\]
for some functions $f_i$, $i\in \{1, \dots, n\}$, on $\mathbb{C}^2$.

It is possible to show \cite{SkrJMP2021}
 that the coordinates of separation $q_i$
 satisfy the following equations:
\begin{gather}\label{abel0}
\sum\limits_{i=1}^n \frac{ {\partial_{I_k} \Phi_i(q_i, p_i, I_1,\dots ,I_n, C_1, \dots ,C_m)} }{ {\partial_{p_i} \Phi_i(q_i, p_i ,I_1, \dots ,I_n, C_1, \dots ,C_m)}}\frac{1}{ f_i(q_i,p_i)} \frac{\partial q_i}{\partial t_j} = \d_{kj}, \qquad \forall k,j\in \{1, \dots, n\},
\end{gather}
where $t_j$ is the ``time'' corresponding to the integral $I_j$, i.e., a parameter along its Hamiltonian flow.

From the equations \eqref{abel0}, one easily deduces the Abel-type equations written in the differential form
\begin{equation}\label{abel0'}
\sum\limits_{i=1}^n \frac{ {\partial_{I_k} \Phi_i(q_i, p_i, I_1, \dots ,I_n, C_1, \dots ,C_m)} }{ {\partial_{p_i} \Phi_i(q_i, p_i ,I_1, \dots ,I_n, C_1, \dots ,C_m)}} \frac{\mathrm{d} q_i}{ f_i(q_i,p_i)} = \mathrm{d} t_k, \qquad k\in \{1, \dots, n\},
\end{equation}
where the momenta $p_i$ satisfy the equations of separation \eqref{sepeq0}.

The equations \eqref{abel0'} for the separated coordinates $q_1, \dots , q_n$ provide a
way to the explicit integration of the equations of motion. So the key problem in the integration process is construction of the coordinates $q_1, \dots, q_n$.
 One of possible methods to do this is the
 method of the vector fields $Z_i$.

\subsection[The method of the vector fields Z\_i]{The method of the vector fields $\boldsymbol{Z_i}$}

The method of the vector field $Z_i$ in the theory of separation
of variables was proposed in the paper
 \cite{MFP} and developed in the papers \cite{FP, FP2}. We will expose the
method in the version convenient for us.

\subsubsection[The vector fields Z\_i]{The vector fields $\boldsymbol{Z_i}$}
Let us assume that we have found a set of separated coordinates
$\{q_1, \dots , q_n, p_1, \dots , p_n\}$ on the generic symplectic leaf
in the Poisson manifold $\mathcal{P}$. The (local) coordinates $q_i$,
$p_j$ are canonical with respect to the brackets $\{ \ , \ \}_1$
 and satisfy some equations of separation \eqref{sepeq0}.

 We consider vector fields $Z_k$ defined in the set of coordinates $\{q_1, \dots , q_n,
p_1, \dots , p_n, C_1, \dots ,\allowbreak C_m\}$ on the Poisson manifold
$\mathcal{P}$ as follows:
\begin{subequations}\label{defZ}
\begin{gather}
Z_k(q_i)=0, \qquad Z_k(p_j)=0, \qquad i,j\in \{1, \dots, n\}, \quad k\in
\{1, \dots, m\},\label{defZ1}
\\
 Z_k(C_l)=\d_{kl}, \qquad k,l\in \{1, \dots ,m \}.\label{defZ2}
\end{gather}
\end{subequations}

 Observe that from the definition of the vector fields
$Z_k$ it immediately follows that
$
 Z_k (Z_l(C_l))\allowbreak=0$, $ i,k,l\in \{1, \dots, m\}$.

In order to proceed with the theory of the vector fields $Z_i$, we
will
 assume that the separated variables $q_i$, $p_j$, $i,j \in
\{1, \dots, n\}$ are the bi-Hamiltonian ones, i.e., on $\mathcal{P}$ there
exists another Poisson structure $\{\ ,\ \}_2$ such that
\begin{equation*}
\{ p_i, q_j\}_2=-q_i\d_{ij}, \qquad \{q_i,q_j\}_2=0,\qquad
\{p_i,p_j\}_2=0, \qquad \forall i,j \in \{1, \dots, n\}.
\end{equation*}

Now we can formulate the following theorem \cite{SkrAsSoVExCle}.
\begin{Theorem}\quad
\begin{itemize}
\item[$(i)$] The vector fields $Z_i$ are symmetries of the Poisson structure
$\{\ , \}_1$, i.e.,
\begin{equation}\label{symZP}
\mathrm{Lie}_{Z_i} \{\ ,\ \}_1=0, \qquad i\in \{1, \dots, m\},
\end{equation}

\item[$(ii)$] The vector fields $Z_i$ satisfy the following conditions with
respect to the Poisson structure~$\{\ , \}_2$:
\begin{equation}\label{symZQ}
\mathrm{Lie}_{Z_i} \{\ ,\ \}_2= \sum\limits_{j=1}^m Z_j \wedge [X_j, Z_i],
\qquad i\in \{1, \dots, m\},
\end{equation}
where $X_j$ is Hamiltonian vector field of $C_j$ with respect
to the second Poisson structure $\{\ ,\ \}_2$.

\item[$(iii)$] If the functions $\Phi_i(q_i,
p_i, I_1, \dots , I_n, C_1, \dots , C_m)$ are linear in $I_j$, $C_r$, $j\in
\{1, \dots, n\}$, $r\in \{1, \dots, m\}$, then
\begin{equation}\label{anHam}
Z_k (Z_l(I_i))=0, \qquad i\in \{1, \dots, n\}, \quad k,l \in \{1, \dots, m\}.
\end{equation}
\end{itemize}
\end{Theorem}

\begin{Remark}
The equalities \eqref{symZP},
\eqref{symZQ}, \eqref{anHam} have been introduced from different considerations
in \cite{MFP}. It was proven later in \cite{FP} that together with
the normalization conditions \eqref{defZ2} they are {\it sufficient} for
the variable separation. This is a basis of the
bi-Hamiltonian method in SoV theory \cite{MFP,FP,FP2}.
\end{Remark}

\begin{Remark}
 Note that the condition \eqref{anHam} implies that the variable separation is of {\it St\"ackel type}, i.e., all equations of separation are {\it linear} in the integrals of motion.
\end{Remark}

\subsubsection{ The separating polynomial}
 Using the above-defined
vector fields $Z_i$, one can define the separating polynomial, whose roots
 are the coordinates of separation. For this purpose, it is
necessary to define the so-called Poisson pencil, i.e., the
linear combination of the brackets $\{\ ,\ \}_1$ and $\{\ ,\
\}_2$,
$
\{\ ,\ \}_u=u \{\ ,\ \}_1+ \{\ ,\ \}_2$.
We will hereafter assume the Gelfand--Zakharevich settings
\cite{GZ}, i.e., we will assume that the Casimir functions of $\{\ ,\ \}_u$ are
polynomial in $u$. Let us denote these Casimir functions by~$C_k(u)$, $k\in \{1, \dots, m\}$. Observe that the functions
$C_k(u)$ are generating functions of the integrals and Casimirs of
$\{\ ,\ \}_1$ and $\{\ ,\ \}_2$, set of integrals
$I_1, \dots , I_n$ is decomposed into subsets~${\{I_{l_1}, I_{l_2},
\dots , I_{l_{n_l}}\}}$, $l\in \{1, \dots, m\}$, in a certain specific way
for each bi-Hamiltonian pair of bracket and enter into the
functions $C_l(u)$ as follows:
\begin{equation}\label{Casu}
C_l(u)=u^{n_l}C_l+ u^{n_l-1} I_{l_1}+ u^{n_l-2} I_{l_2}+ \dots +
I_{l_{n_l}}, \qquad l\in \{1,\dots, m\},
\end{equation}
so that $n_1+n_2+ \dots +n_m=n$ and $I_{l_{n_l}}$ is a Casimir of $\{\
,\ \}_2$.

The following theorem holds true \cite{MFP,FP,FP2}.
\begin{Theorem}\label{MFP}
Let the vector fields $Z_i$ satisfy the conditions \eqref{symZP}, \eqref{symZQ}, \eqref{anHam} and normalization conditions
\eqref{defZ2}.
Let the roots $u=q_i$, $i\in \{1,\dots, n\}$, of the equation
\[
S(u)=\mathrm{det}(Z_i(C_j(u)))=0, \qquad i,j \in \{1,\dots, m\},
\]
be functionally independent on generic symplectic leaves of $\{\ ,\ \}_1$. Then $q_i$, $i\in \{1, \dots, n\}$, are the coordinates of separation for the considered
bi-Hamiltonian system.
\end{Theorem}

\subsection[The algebraic vector field Z: the non-St\"ackel case]{The algebraic vector field $\boldsymbol{Z}$: the non-St\"ackel case}
\subsubsection[The algebraic vector field Z]{The algebraic vector field $\boldsymbol{Z}$}
In the case of bi-Hamiltonian SoV, the integrals and Casimir functions enter into the equations of separation via the Casimirs $C_i(u)$ , $i\in \{1,\dots, m\}$ (see formula \eqref{Casu}) of the Poisson pencil~$\{\ , \}_u$
\[
\Phi_{i}(q_i,p_i, I_1, \dots ,I_n, C_1, \dots ,C_m)=\Phi_{i}(q_i,p_i, C_1(q_i), \dots ,C_m(q_i))=0, \qquad i\in \{1,\dots, n\}.
\]
The non-St\"ackel SoV means in this context that certain Casimirs of the Poisson pencil, say~$C_i(q_i)$, $i\in \{1,\dots, r\}$ enter into $\Phi_{i}(q_i,p_i, C_1(q_i), \dots ,C_m(q_i))$ in a nonlinear way. We call these Casimirs to be ``non-St\"ackel''.

The following proposition holds true \cite{SkrAsSoVExCle}.

 \begin{Proposition}\label{Z}
 Let $q_i$, $p_i$, $i\in \{1,\dots, n\}$, be separated coordinates. Let the corresponding equations of separation be nonlinear in the Casimir functions of the Poisson pencil $C_1(q_i),\dots ,\allowbreak C_r(q_i)$ and linear in all other integrals and Casimir functions. Let $Z$ be an invariance vector field: $Z(q_i)=Z(p_i)=0$, $i\in \{1,\dots, n\}$. Let the following conditions also hold true:
 \begin{equation}\label{condZ}
 Z(C_1(u))= \dots =Z(C_r(u))=0.
 \end{equation}

 Then
\begin{itemize}\itemsep=0pt
\item[$(i)$] If
 $
n_1+n_2+ \dots +n_r=m-r-1$,
then the vector field $Z$ is unambiguously defined in terms of the vector fields $Z_i$ by the conditions \eqref{condZ} and the normalization condition $Z(C_m)=1$.

\item[$(ii)$] If, moreover, $Z(Z(C_{r+1}))=\dots =Z(Z(C_m))=0$, then the square of the vector field $Z$ annihilates all the integrals and Casimir functions
 \begin{equation}\label{quadZ1}
 Z(Z(C_1(u)))= \dots =Z(Z(C_m(u)))=0.
 \end{equation}
 \end{itemize}
 \end{Proposition}

We need also to formulate the following important conjecture.
\begin{Conjecture}\label{conj2} Let the space $\mathcal{P}$ coincides with dual space of a Lie algebra $\mathfrak{g}$ and Poisson brackets $\{\ ,\ \}_1$ coincide with a standard Lie--Poisson brackets on $\mathfrak{g}^*$.
Then among the invariance vector fields, i.e., vector fields such that $Z(q_i)=Z(p_i)=0$, $i\in \{1, \dots, n\}$, there exists a~special vector field $Z$
 which, acting in the system of the natural Lie-algebraic coordinates $($linear coordinate functions of the space $\mathfrak{g}^*)$ annihilates its own components.
\end{Conjecture}

Hereafter, we will naturally assume that the vector field $Z$ described in Proposition \ref{Z} and the vector field $Z$ from the Conjecture \ref{conj2} are the same vector fields. This will transform the differential equations~\eqref{quadZ1} for the components of the vector field $Z$, to the system of quadratic algebraic equations for them, i.e., will give a possibility to reduce the problem of the construction of the vector field $Z$ in terms of the initial variables to the problem of solving system of {\it algebraic} equations for its components in the initial system of coordinates. We will call such the vector field~$Z$ {\it algebraic}.
Let us now consider a dimensional constraint that is imposed by Proposition~\ref{Z} together with the condition that the vector field $Z$ is algebraic.
The following proposition holds true.
\begin{Proposition}
For the algebraic vector fields $Z$, existence of the generic solution of the equation \eqref{quadZ1} satisfying also the condition \eqref{condZ} implies that the number of Casimir functions~$m$ is equal to the number of the integrals of motion $n$, i.e., $m=n$.
\end{Proposition}

\begin{proof}
 The number of components of the vector field $Z$ in the initial coordinate system is equal to $2n+m$. The number of equations \eqref{quadZ1} is equal to $n+m$. That is why the generic solution of the system of equations \eqref{quadZ1} is $n$-parametric. On the other hand, the algebraic equations \eqref{quadZ1} and \eqref{condZ} are homogeneous. That is why there is one common multiplicative parameter among $n$ parameters of generic solution of \eqref{quadZ1} not fixed by \eqref{condZ}. It is fixed by the normalization condition $Z(C_m)=1$. All other $n-1$ parameters should be fixed by the constraints \eqref{condZ}. That is why the number of the constraints \eqref{condZ} should be equal to $n-1$, i.e., $(n_1+1)+ (n_2+1)+ \dots +(n_r+1)=n-1$. On the other hand, the Proposition \ref{Z} provides the following dimensional constraint: $n_1+n_2+ \dots +n_r=m-r-1$.
Comparing these two-dimensional constraints, we immediately obtain that $m=n$.
\end{proof}

\begin{Remark}
 Observe, that although the condition $m=n$ is a restrictive one, it is a typical situation for the Lax-integrable systems in the small-rank cases. To generalize this condition in the higher rank cases, one will probably need to consider several vectors fields $Z$ with the above special properties. Let us remark, that the condition $m=n$ has appeared in another context in the recent paper \cite{BlaMar}.
\end{Remark}

\subsubsection{The separating polynomial: the non-St\"ackel case}
Having found the algebraic vector field $Z$ satisfying \eqref{condZ}, \eqref{quadZ1} we can start to look for the separating polynomial. Unfortunately, in the general $m>1$ case, one cannot construct the formula for the separating polynomial with the help of only one vector field $Z$. One can proceed in a systematic way only for certain types of the equations
of separation. That is why we will {\it assume} the following form of the equations of separation:
\begin{gather}
\Phi_{i}(q_i,p_i, C_1(q_i),\dots ,C_m(q_i))\nonumber\\
\qquad{} = \Phi_{i}(q_i,p_i, C_1(q_i),\dots ,C_r(q_i), \phi(q_i,C_{r+1}(q_i),\dots ,C_m(q_i)))
 =0,\label{eqsepbhmf}
\end{gather}
where function $\phi(u,C_{r+1}(u), \dots ,C_m(u))$ is linear in $C_{s}(u)$, $s\in \{r+1,\dots, m\}$, the functions $\Phi_{i}$, $i\in \{1, \dots,n \}$ are linear in $\phi(u,C_{r+1}(u), \dots ,C_m(u))$ and nonlinear in $C_k(q_i)$, $k\in \{1, \dots, r \}$.

Acting on the equations \eqref{eqsepbhmf} by the vector field $Z$ and taking into account that by our construction $Z(q_1)= \dots =Z(q_n)= Z(p_1)= \dots =Z(p_n)=0$, $Z(C_1(u))= \dots =Z(C_r(u))=0$,
we easily obtain that the coordinates $q_i$ should satisfy the following equations:
\begin{equation}\label{seppol}
Z(\phi(q_i,C_{r+1}(q_i), \dots ,C_m(q_i)))=0, \qquad i\in \{1, \dots, n\},
\end{equation}
i.e., that $S(u)=Z(\phi(u,C_{r+1}(u), \dots ,C_m(u)))$ is a separating polynomial in $u$ of degree $n$.

\begin{Remark}
In practice, we do not know a priori the equations of separation \eqref{eqsepbhmf}. That is why we have to {\it construct} the function $\phi$ as polynomial of degree $n$ in $u$ by the linear combination of the polynomials $C_{s}(u)$, $s\in \{r+1, \dots, m \}$ of degree $n_s$ in $u$, where
 $n_1+n_2+ \dots +n_r=m-r-1$, $n_1+n_2+ \dots +n_m=n$, with monomial in $u$ coefficients.
The form of separating polynomial~\eqref{seppol} will be only indicative. Indeed, since we have no {\it closed} differential conditions of the type \eqref{symZP}--\eqref{symZQ} for the {\it unique} vector field $Z$, we have yet to control that the roots $q_i$ of $S(u)$ Poisson-commute and that $q_i$ together with the corresponding conjugated momenta $p_i$ satisfy some equations of separation of the form \eqref{eqsepbhmf}. In the next sections, we will illustrate this on the examples of the extended Clebsh and Manakov models.
\end{Remark}

\section{Three-dimensional extension of the Clebsch model}\label{sec3}
\subsection{The model}
Let us consider nine-dimensional linear space with the
coordinates $S_{\a}$, $T_{\a}$, $W_{\a}$, $\a\in \{1,2,3\}$, which satisfy the
 following Lie--Poisson brackets:
 \begin{subequations}\label{e3lpb}
\begin{gather}
\{ S_{\a}, S_{\b}\}_1=\e_{\a\b\g}(S_{\g}+ j_{\g} W_{\g}), \qquad \{ S_{\a},
T_{\b}\}_1=\e_{\a\b\g}T_{\g},\qquad \{ S_{\a}, W_{\b}\}_1= \e_{\a\b\g} W_{\g},\label{e3lpb1}
\\
\{ T_{\a}, T_{\b}\}_1=\e_{\a\b\g} W_{\g}, \qquad \{ T_{\a},
W_{\b}\}_1=0,\qquad \{ W_{\a}, W_{\b}\}_1= 0,\label{e3lpb2}
\end{gather}
\end{subequations}
where $j_{\a}$, $\a\in \{1,2,3\}$, are arbitrary parameters, and we will hereafter assume
 that $j_{\a}\neq j_{\b}$ if~${\a\neq \b}$.

These brackets possess three Casimir functions
\begin{equation}\label{casimir}
C_3= \sum_{\a=1}^3 W^2_{\a},\qquad C_2= \sum_{\a=1}^3 \big(j_{\a} W^2_{\a} + 2 W_{\a}S_{\a} +T_{\a}^2 \big) \qquad C_1= 2\sum_{\a=1}^3T_{\a}W_{\a}.
\end{equation}
Let us consider the following Hamiltonian:
\begin{equation}\label{h}
H =\sum\limits_{\a=1}^3 \big(S_{\a}^2+ (j_{\b}+j_{\g}) T^2_{\a}+ 2 j_{\a} S_{\a} W_{\a}\big),
\end{equation}
where the indices $\a$, $\b$, $\g$ hereafter denote cyclic permutation of the indices $1$, $2$, $3$, i.e., for $\a=1$ we have that $\b=2$, $\g=3$, for $\a=2$ we have that $\b=3$, $\g=1$, for $\a=3$ we have that $\b=1$, $\g=2$.

As it is easy to check, it possesses two quadratic integrals
 of the following explicit form:
\begin{gather}
K= \sum\limits_{\a=1}^3 \big(j_{\a} S_{\a}^2+ j_{\b}j_{\g} T^2_{\a}\big), \label{k}
\\
L=2\sum_{\a=1}^3T_{\a}S_{\a}.\label{l}
\end{gather}
The model is integrable in the sense of Liouville: the dimension
of the generic symplectic leaf (level set of three Casimir
functions) is six and we have three independent Poisson commuting integrals
$\{H,K\}_1=0$, $ \{H,L\}_1=0$, $ \{K,L\}_1=0$
which is sufficient for the complete integrability of the
considered Hamiltonian system. Observe, that from the point of
view of the notations of the previous section we have that $n=m=3$ in this case.

\begin{Remark}
 Note that the Poisson algebra \eqref{e3lpb} possesses an ideal spanned by $\{W_{\a}\mid \a\in \{1,2,3\} \}$. Factorizing over this ideal, i.e., putting $W_{\a}=0,$ $\a\in \{1,2,3\}$, we obtain that the Lie--Poisson brackets \eqref{e3lpb} become the standard Lie--Poisson brackets on $\mathfrak{e}^*(3)$, the Casimir functions~$C_3$ and $C_1$ turn zero, the Casimir $C_2$ and integral $L$ become Casimir functions on~$\mathfrak{e}^*(3)$ and the integrals $H$ and $K$ become the Hamiltonian and integral of motion of the famous Clebsch model \cite{Clebsch}. That is why we call the integrable system considered in the present section to be {\it three-dimensional extension of the Clebsch model}.
\end{Remark}

\subsection{The bi-Hamiltonian structure and the Poisson pencil}
It is possible to show that there is a second Poisson structure
for our extended Clebsch system, compatible with the first one, standing behind the integrability of the model and having
the form
 \begin{subequations}\label{so4lpb}
\begin{gather}
\{ S_{\a}, S_{\b}\}_2=\e_{\a\b\g} j_{\g} S_{\g}, \qquad \{ S_{\a},
T_{\b}\}_2=\e_{\a\b\g} j_{\b} T_{\g},\qquad \{ T_{\a}, T_{\b}\}_2= \e_{\a\b\g} S_{\g},\label{e4lpb1}
\\
\{ T_{\a}, W_{\b}\}_2=0, \qquad \{ S_{\a},
W_{\b}\}_2=0,\qquad \{ W_{\a}, W_{\b}\}_2= - \e_{\a\b\g} W_{\g}.\label{so4lpb2}
\end{gather}
\end{subequations}
The function $C_3$ is a common Casimir function
for the both brackets $\{\ , \ \}_1$ and $\{\ , \ \}_2$, the other two Casimir functions of the brackets $\{\ , \ \}_2$ are the functions $K$ and $L$. The functions $H$, $C_1$ and $C_2$ are the Poisson-commuting integrals of motion with respect to $\{\ , \ \}_2$
$\{H,C_1\}_2=0$, $ \{H,C_2\}_2=0$, $ \{C_1,C_2\}_2=0$.
Due to the compatibility of the brackets $\{\ , \ \}_1$ and $\{\ , \ \}_2$, it is possible to consider the so called Poisson pencil of the brackets $\{\ , \ \}_1$ and $\{\ , \ \}_2$,
$
\{\ , \ \}_u= u\{\ , \ \}_1+ \{\ , \ \}_2$.
The function $C_3$ is a~Casimir
of $\{\ , \ \}_u$. The second and third Casimirs of $\{\ , \ \}_u$ are the functions
$
C_2(u)= u^{2}C_2+u H+K$,
$
C_1(u)= u C_1+L$.
They will be used below while constructing separating polynomial and equations of separation. We will also use the following cubic in $u$ combination of the Casimirs $C_3(u)\equiv C_3$ and $C_2(u)$:
\[
\phi(u,C_2(u),C_3)= u^3 C_3 +C_2(u)=u^3 C_3+ u^{2}C_2+u H+K.
\]

\begin{Remark}
 Observe that in the case $j_{\a}\neq 0$, $\a\in \{1,2, 3\}$, the Poisson algebra \eqref{so4lpb} is isomorphic to $\mathfrak{so}(4)\oplus \mathfrak{so}(3)$. The isomorphism is achieved by the following substitution of variables
\smash{$S_{\a}\rightarrow \sqrt{j_{\b}} \sqrt{j_{\g}} S_{\a}$}, \smash{$ T_{\a}\rightarrow \sqrt{j_{\a}} T_{\a}$}, $ W_{\a}\rightarrow - W_{\a}$, $ \a\in \{1,2,3\}$.
In such a way, the considered model is isomorphic to the so-called Manakov or Schottky--Frahm model on $\mathfrak{so}(4)$ interacting with anisotropic $\mathfrak{so}(3)$ Euler top.
Due to the isomorphism $\mathfrak{so}(4)\simeq \mathfrak{so}(3)\oplus \mathfrak{so}(3)$, the Poisson algebra \eqref{so4lpb} is isomorphic to $\mathfrak{so}(3)\oplus \mathfrak{so}(3)\oplus \mathfrak{so}(3)$ and the corresponding integrable system is equivalent to the system of three interacting anisotropic tops. It is possible to show that this system is equivalent to $N=3$ elliptic Gaudin model, but we will not develop this line in the present article.
\end{Remark}

\section[The algebraic vector field Z]{The algebraic vector field $\boldsymbol{Z}$}\label{sec4}
In this section, we will solve the equations
\eqref{quadZ1}
\begin{alignat}{4}
&Z(Z(H))=0,\qquad &&Z(Z(K))=0, \qquad&& Z(Z(L))=0,&\nonumber\\ &Z(Z(C_1))=0, \qquad&& Z(Z(C_2))=0, \qquad&& Z(Z(C_3))=0&\label{annul}
\end{alignat}
under the additional assumption that the vector field $Z$ is algebraic.

The general vector field on the three-dimensional extension of $\mathfrak{e}^*(3)$ is written as follows:
\[
Z=\sum\limits_{\a=1}^3\left( A_{\a}\frac{\partial}{ \partial S_{\a}}+
B_{\a}\frac{\partial}{ \partial T_{\a}}+ D_{\a}\frac{\partial}{ \partial W_{\a}}\right),
\]
where $A_{\a}$, $B_{\a}$, $D_{\a}$, $\a\in \{1,2,3\}$, are certain
functions on $\mathfrak{e}^*(3)$.

In order to solve the equations \eqref{annul}, we will impose
 an additional restriction that vector field~$Z$ is algebraic, i.e., annihilates its own components
\begin{equation}\label{ZAB}
Z(A_{\a})=0, \qquad Z(B_{\a})=0, \qquad Z(D_{\a})=0, \qquad \a\in \{1,2,3\}.
\end{equation}
This condition transforms the system of the differential equations \eqref{annul}
for the vector field $Z$ in the system of the
algebraic equations for its components $A_{\a}$, $B_{\a}$, $D_{\a}$, $\a\in \{1,2,3\}$.

For the future, we need to introduce the following ``elliptic'' constants $c_{\a}$:
\begin{equation}\label{calfa}
c^2_{\a}=j_{\b}-j_{\g}, \qquad \a\in \{1,2,3\}.
\end{equation}

The following proposition holds true.
\begin{Proposition}
The vector field $Z$ with the following components:
\begin{subequations}\label{A-B}
\begin{gather}
A_{\a} =\l c_{\a} \big(x^2+(j_{\b}+j_{\g}-j_{\a}) y^2+2 j_{\b} j_{\g} y + j_{\a} j_{\b} j_{\g}\big),\label{Ai}
\\
B_{\a} = 2 \l c_{\a} x (y+j_{\a}),\label{Bi}
\\
D_{\a} = \l c_{\a} \big(y^2+2j_{\a} y+ j_{\a}(j_{\b}+j_{\g})-j_{\b}j_{\g}\big) ,\label{Di}
\end{gather}
\end{subequations}
where $\l$, $x$ and $y$ are arbitrary functions annihilated
by $Z$, i.e., $Z(x)=0$, $Z(y)=0$, $Z(\l)=0$, is a solution of the system
of equations \eqref{annul} possessing the property \eqref{ZAB}.
\end{Proposition}

\begin{proof}
 In order to prove the proposition, we rewrite the equations \eqref{annul} (with the help of the conditions \eqref{ZAB}) in the form of six algebraic equations
\begin{subequations}\label{HKL}
\begin{gather}
\sum\limits_{\a=1}^3 \big(A_{\a}^2+ (j_{\b}+j_{\g}) B^2_{\a}+ 2 j_{\a} A_{\a} D_{\a}\big)=0, \qquad
 \sum\limits_{\a=1}^3 \big(j_{\a} A_{\a}^2+ j_{\b}j_{\g} B^2_{\a}\big)=0, \label{H}\\ \sum_{\a=1}^3B_{\a}A_{\a}=0,
\qquad
\sum_{\a=1}^3B_{\a}D_{\a}=0,\\
 \sum_{\a=1}^3 \big(j_{\a} D^2_{\a} + 2 D_{\a} A_{\a} +B_{\a}^2 \big)=0, \qquad \sum_{\a=1}^3 D^2_{\a}=0.\label{C3}
\end{gather}
\end{subequations}
Substituting the ansatz \eqref{A-B} into the equations \eqref{HKL}, taking into account the definitions
\eqref{calfa} and the following equalities
$
\sum_{\a=1}^3 c_{\a}^2=0$, $ \sum_{\a=1}^3 j_{\a} c_{\a}^2=0$,
after the direct calculations, we obtain that the equations \eqref{HKL} hold true.
\end{proof}

\begin{Remark}
Observe that we have six independent equations \eqref{HKL} 
 for nine functions $A_{\a}$, $B_{\a}$, $D_{\a}$, $\a \in \{1,2,3\}$. That is why the generic solution of the equations \eqref{HKL} 
 is three parametric. In the case of the proposed solution, these three parameters are $\l$, $x$ and $y$. We conjecture that there are only two three-parametric, non-equivalent solutions of the equations \eqref{annul} satisfying also the conditions \eqref{ZAB}. They are the presented solution \eqref{A-B} and the solution found in our paper \cite{SkrAsSoVExCle}.
\end{Remark}

Now we have to define the functions $\l$, $x$ and $y$. As it follows from our theory of non-St\"ackel vector field $Z$, it should annihilate one of the non-common Casimir functions $C_2(u)$ or $C_1(u)$ of~$\{\ ,\ \}_u$. It is easy to see that only the Casimir $C_1(u)$ can be taken for such a role because the condition $n_1=m-r-1$ is satisfied for it: we have $m=3$, $r=1$, $n_1=1$. That is why we will impose the condition
$Z(C_1(u))=0$,
 that is, we will define the functions $x$ and $y$ from the following two equations:
\[
Z(C_1)=\sum\limits_{\a=1}^3 (T_{\a} D_{\a}+ W_{\a} B_{\a})= 0, \qquad Z(L)= \sum\limits_{\a=1}^3 (T_{\a} A_{\a}+ S_{\a} B_{\a})= 0.
\]
These will be our main constraint equations. They are written more explicitly as follows:
\begin{subequations}\label{xy}
\begin{gather}
\sum\limits_{\a=1}^3 c_{\a}\bigl(y^2+2j_{\a} y+j_{\a}(j_{\b}+j_{\g})-j_{\b}j_{\g}) T_{\a}+ 2 x (y+j_{\a}) W_{\a} \bigr)= 0,
\\
\sum\limits_{\a=1}^3 c_{\a}\bigl( \big(x^2+ (j_{\b}+j_{\g}-j_{\a}) y^2+2 j_{\b} j_{\g} y + j_{\a} j_{\b} j_{\g}\big) T_{\a}+ 2 x (y+j_{\a}) S_{\a} \bigr)= 0.
\end{gather}
\end{subequations}

The following proposition holds true.
\begin{Proposition}\label{Zxynul}
Let the functions $x$ and $y$ be solutions of the equations \eqref{xy}. Then they are annihilated by the vector field $Z$,
\begin{equation}\label{annulxy}
Z(x)=0, \qquad Z(y)=0.
\end{equation}
\end{Proposition}
 \begin{proof}
 In order to prove the proposition, we use that by the very same definition
\begin{gather*}
Z(x)= \sum\limits_{\a=1}^3\left( A_{\a}\frac{\partial x}{ \partial S_{\a}}+
B_{\a}\frac{\partial x}{ \partial T_{\a}}+ D_{\a}\frac{\partial x}{ \partial W_{\a}}\right), \\
Z(y)= \sum\limits_{\a=1}^3\left( A_{\a}\frac{\partial y}{ \partial S_{\a}}+
B_{\a}\frac{\partial y}{ \partial T_{\a}}+ D_{\a}\frac{\partial y}{ \partial W_{\a}}\right).
\end{gather*}
To prove that the above expressions are zero, it is necessary to find explicitly the derivatives~\smash{$\frac{\partial x}{\partial Y_{ \a}}$} and $\frac{\partial y}{\partial Y_{ \a}}$, where $Y_{\a}=S_{\a}$ or $Y_{\a}=T_{\a}$ or $Y_{\a}=W_{\a}$.
 This is done using the equations \eqref{xy}. In more details, for each $\a\in \{1,2,3\}$, we have the following system of two equations that permits to find $\frac{\partial x}{\partial Y_{ \a}}$, $\frac{\partial y}{\partial Y_{ \a}}$:
 \begin{gather*}
 \frac{\partial Z(C_1)}{\partial y} \frac{\partial y}{\partial Y_{ \a}}+ \frac{\partial Z(C_1)}{\partial x} \frac{\partial x}{\partial Y_{ \a}}+ \frac{\partial Z(C_1)}{\partial Y_{ \a}}=0,
\\
 \frac{\partial Z(L)}{\partial y} \frac{\partial y}{\partial Y_{ \a}}+ \frac{\partial Z(L)}{\partial x} \frac{\partial x}{\partial Y_{ \a}}+ \frac{\partial Z(L)}{\partial Y_{ \a}}=0.
 \end{gather*}
 They are obtained by the differentiation of the constraints $Z(C_1)=0$ and $Z(L)=0$ with respect to $Y_{ \a}$.

Now, using the derivatives $\frac{\partial x}{\partial Y_{ \a}}$ and $\frac{\partial y}{\partial Y_{ \a}}$ calculated as it is explained above, the explicit form of the components $A_{\a}$, $B_{\a}$, $D_{\a}$ given by the formulae \eqref{A-B}, the explicit form of the coordinates and momenta of separation given in the theorem, the constraint equations \eqref{xy}, the definition~\eqref{calfa} of the constants $c_{\a}$, after direct and tedious calculations, we obtain the equalities \eqref{annulxy}.
\end{proof}

\begin{Remark}
The overall multiplier $\l$ is determined from the normalization condition $Z(C_3)=1$, i.e., $\l=\bigl(\sum_{\a=1}^3 D_{\a} W_{\a})^{-1}$. Using the fact that $Z(x)=Z(y)=0$ (see the proposition above), it is easy to show that $Z(\l)=0$, i.e., the construction of this section is self-consistent.
\end{Remark}
\section{Symmetric separated variables}\label{sec5}
\subsection{Separating polynomial and the coordinates of separation}
\subsubsection{The separating polynomial}
As it follows from the exposed above bi-Hamiltonian theory of non-St\"ackel vector field $Z$ one should look for the coordinates of separation as the roots of the following cubic polynomial:
\begin{align}
S(u)&= Z(\phi(u,C_2(u),C_3))= Z\big(u^3 C_3+C_2(u)\big)\nonumber\\
&= Z(C_3) u^3+ Z(C_2) u^2+ Z(H) u +Z(K),\label{sepp}
\end{align}
where $u$ is a parameter of the bi-Hamiltonian pencil and the needed vector field $Z$ is defined as in the previous section, i.e., $Z$ is algebraic and $Z(C_1(u))=0$.

\begin{Remark}
 Observe that the needed roots of the polynomial \eqref{sepp} do not depend on the overall normalization coefficient $\l$, that is why we will simply ignore it, putting hereafter that~${\l=1}$.
 \end{Remark}

\subsubsection{The auxiliary coordinate system}
In order to better understand the structure of $S(u)$ and simplify the constraint \eqref{xy}, we will introduce a system of the auxiliary coordinates. In more details, we will consider the following set of nine functions:
\begin{subequations}\label{fgh}
\begin{alignat}{4}
 &f_1 = \sum\limits_{\a=1}^3 D_{\a} S_{\a},\qquad&&
 f_2 = \sum\limits_{\a=1}^3 B_{\a} S_{\a}, \qquad&&
 f_3 = \sum\limits_{\a=1}^3 A_{\a} S_{\a},&
\\
& g_1 = \sum\limits_{\a=1}^3 D_{\a} T_{\a},\qquad&&
 g_2 = \sum\limits_{\a=1}^3 B_{\a} T_{\a}, \qquad&&
 g_3 = \sum\limits_{\a=1}^3 A_{\a} T_{\a},&
\\
& h_1 = \sum\limits_{\a=1}^3 D_{\a} W_{\a},\qquad&&
 h_2 = \sum\limits_{\a=1}^3 B_{\a} W_{\a}, \qquad&&
 h_3 = \sum\limits_{\a=1}^3 A_{\a} W_{\a}.&
 \end{alignat}
\end{subequations}
It is easy to see, that in the terms of these functions the constraint equations \eqref{xy} are written as follows:
\begin{equation}\label{mconstr}
h_2+g_1=0,\qquad f_2 +g_3=0.
\end{equation}
In other words, as a system of (local) coordinates on the nine-dimensional phase space one can take nine functions $f_1$, $f_3$, $g_2$, $g_3$, $h_1$, $h_2$ $h_3$, $x$, $y$. It is possible to show that they are functionally independent.

The formulae \eqref{fgh} define the invertible quasi-linear, i.e., linear in ($S_{\a}, T_{\a}, W_{\a}$) and $(f_{\a}, g_{\a},\allowbreak h_{\a})$ but nonlinear in $x$, $y$, map $(S_{\a}, T_{\a}, W_{\a}) \rightarrow (f_{\a}, g_{\a}, h_{\a})$. We have
\begin{subequations}\label{STW}
\begin{gather}
 S_{\a}=\frac{c_{\a}}{2 c_1^2c_2^2c_3^2 }\left( \left(1+\frac{y \big(y^2+2 y j_{\a}+j_{\a} (j_{\b}+j_{\g})-j_{\b} j_{\g}\big)}{x^2+(y+j_{1}) (y+j_2) (y+j_3)}\right) f_1-
 \frac{(y+j_{\a}) f_2}{x}\right.\nonumber\\
\phantom{S_{\a}=}{}
\left.+\frac{\big(y^2+2 y j_{\a}+j_{\a} (j_{\b}+j_{\g})-j_{\b} j_{\g}\big) f_3}{x^2+(y+j_1) (y+j_2) (y+j_3)}
 \right),
\\
 T_{\a}= \frac{c_{\a}}{2 c_1^2c_2^2c_3^2 }\left( \left(1+\frac{y \big(y^2+2 y j_{\a}+j_{\a} (j_{\b}+j_{\g})-j_{\b} j_{\g}\big)}{x^2+(y+j_{1}) (y+j_2) (y+j_3)}\right) g_1\right.\nonumber\\
\phantom{T_{\a}= }{}
\left.- \frac{(y+j_{\a}) g_2}{x}+\frac{\big(y^2+2 y j_{\a}+j_{\a} (j_{\b}+j_{\g})-j_{\b} j_{\g}\big) g_3}{x^2+(y+j_1) (y+j_2) (y+j_3)}
 \right),
\\
 W_{\a}=\frac{c_{\a}}{2 c_1^2c_2^2c_3^2 }\left( \left(1+\frac{y \big(y^2+2 y j_{\a}+j_{\a} (j_{\b}+j_{\g})-j_{\b} j_{\g}\big)}{x^2+(y+j_{1}) (y+j_2) (y+j_3)}\right) h_1\right.\nonumber\\
 \phantom{W_{\a}=}{}
 \left.-\frac{(y+j_{\a}) h_2}{x}+\frac{\big(y^2+2 y j_{\a}+j_{\a} (j_{\b}+j_{\g})-j_{\b} j_{\g}\big) h_3}{x^2+(y+j_1) (y+j_2) (y+j_3)}
 \right).
 \end{gather}
\end{subequations}
We use these formulae in the subsequent proofs.

\subsubsection[The coordinates of separation q\_1, q\_2, q\_3]{The coordinates of separation $\boldsymbol{q_1}$, $\boldsymbol{q_2}$, $\boldsymbol{q_3}$}
By direct calculation, using the formulae \eqref{STW}, it is possible to show that on the constraint \eqref{mconstr} the polynomial $S(u)$ acquires (in terms of the introduced $f-g-h$ coordinates) the following form:
\begin{equation}\label{Su}
S(u)= u^3 s_3+ u^2 s_2+ u s_1+ s_0,
\end{equation}
where
\begin{subequations}\label{Si}
\begin{gather}
s_3= 2 x h_1\bigl(x^2+(y+j_1)(y+j_2)(y+j_3)\bigr),
\\
s_2= 2(f_1+g_2+h_3 -y h_1) x^3 + h_2 \big(3 y^2+2 y (j_1+j_2+j_3)+j_1 j_2+j_1 j_3+j_2 j_3\big) x^2 \nonumber\\
\phantom{s_2=}{}+
 (y+j_1)(y+j_3)(y+j_2)\big(2(f_1+g_2-h_3-3 y h_1)x \nonumber\\
\phantom{s_2=}{}+ h_2\big(3y^2+2y (j_1+j_2+j_3)+j_1 j_2+j_1j_3+j_2j_3\big)\big),
\\
s_1= -x^4 h_2+2( f_3 -(f_1+ h_3+2 g_2)y) x^3- \big(3 y^2+2 y (j_1+j_2+j_3)+j_1 j_2+j_1 j_3+j_2 j_3\big)\nonumber
\\
\phantom{s_1=}{}\times(3 y h_2+f_2) x^2+(y+j_1) (y+j_2) (y+j_3)
\bigl(-2\big(-2y^2 h_1+(-h_3+3f_1+2g_2)y\nonumber
\\
\phantom{s_1=\times}{}+f_3 \big)x -2y^3h_2+(3f_2-h_2(j_1+j_3+j_2))y^2
+2(j_1+j_3+j_2)f_2y\nonumber
\\
\phantom{s_1=\times}{}+(j_1j_2+j_1j_3+j_2j_3)f_2+j_1 j_2 j_3 h_2\bigr),
\\
s_0= (f_2+2 y h_2) x^4-2 y (f_3-y g_2) x^3+ \bigl( 4y^4 h_2+(f_2+2h_2(j_1+j_2+j_3))y^3\nonumber
\\
\phantom{s_1=}{} -((j_1j_3+j_2j_3+j_1j_2)f_2
+2j_1j_2j_3 h_2)y+2j_1j_2j_3 f_2\bigr) x^2\nonumber
\\
\phantom{s_1=}{} + (y+j_1)(y+j_2)(y+j_3)\bigl(2 y((2f_1+g_2) y +f_3)x \nonumber
\\
\phantom{s_1=}{}
 -f_2 \big(2 y^3+y^2( j_1+j_2+j_3)-j_1 j_2 j_3\big)\bigr),
\end{gather}
\end{subequations}
and we have taken into account the constraints \eqref{mconstr}.

Hence, its roots $u=q_1$, $u=q_2$ and $u=q_3$
are our candidates for the separated coordinates. In order to prove that they are the coordinates of separation indeed, in the next subsection we show that they Poisson-commute, find the corresponding canonically conjugated momenta and equations of separation.

\subsection[The momenta of separation p\_1, p\_2, p\_3]{The momenta of separation $\boldsymbol{p_1}$, $\boldsymbol{p_2}$, $\boldsymbol{p_3}$}
In order to construct the momenta of separation $p_1$, $p_2$, $p_3$ (providing that $q_1$, $q_2$ and $q_3$ are the coordinates of separation indeed), we will proceed directly finding the momenta from the condition that they are (quasi)canonically conjugated to the coordinates $q_1$, $q_2$ and $q_3$.

The following theorem holds true.
\begin{Theorem}\label{momenti}
Let the functions $x$ and $y$ be solutions of the equations \eqref{xy}. Let the components~$A_{\a}$, $B_{\a}$, $D_{\a}$ of the vector field $Z$ be defined by the formulae \eqref{A-B} and the functions $f_{\a}$, $g_{\a}$, $h_{\a}$ be defined by the formulae \eqref{fgh}. Then
\begin{itemize}\itemsep=0pt
\item[$(i)$] The coordinates $q_1$, $q_2$, $q_3$ defined in the previous subsection Poisson-commute with respect to the both brackets $ \{\ ,\ \}_1$ and $\{\ , \ \}_2$
\begin{equation}\label{qcanon}
\{q_i,q_j\}_1=0, \qquad
\{q_i,q_j\}_2=0, \qquad i, j\in \{1,2,3\}.
\end{equation}

\item[$(ii)$] The functions $p_i$, $i \in \{1,2,3\}$ given by the following formulae:
\begin{equation}\label{pi}
 p_i=p(u)|_{u=q_i},
 \end{equation}
 where
\begin{align}
p(u)= {}&
\frac{1}{2 c_1 c_2 c_3 }\frac{\big((u-2y) x^2+u(y+j_1)(y+j_2)(y+j_3)\big) h_2}{x\big(x^2+(y+j_1)(y+j_2)(y+j_3)\big)(u-y)}\nonumber\\
&+\frac{\big((y+j_1)(y+j_2)(y+j_3)-x^2\big) f_2}{x\big(x^2+(y+j_1)(y+j_2)(y+j_3)\big)(u-y)},\label{pu}
\end{align}
are their canonically conjugated momenta with respect to the brackets $\{ \ ,\ \}_1$, i.e.,
\begin{subequations} \label{qcanon1}
\begin{gather}
 \{p_i,q_j\}_1= \d_{ij}, \qquad i, j\in \{1,2,3\},
\\
\{p_i,p_j\}_1=0, \qquad i, j\in \{1,2,3\},
\end{gather}
\end{subequations}
and quasi-canonically conjugated momenta with respect to the brackets $\{ \ ,\ \}_2$, i.e.,
\begin{subequations}\label{qcanon2}
\begin{gather}
 \{p_i,q_j\}_2=-q_i \d_{ij}, \qquad i, j\in \{1,2,3\},
\\
\{p_i,p_j\}_2=0, \qquad i, j\in \{1,2,3\}.
\end{gather}
\end{subequations}
\item[$(iii)$] The variables $q_i$, $p_j$ are $Z$-invariants, i.e.,
\begin{equation}\label{qpinv}
Z(q_i)=Z(p_j)=0, \qquad i,j \in \{1,2,3\}.
\end{equation}
\end{itemize}
\end{Theorem}
\begin{proof}
The proof of the items $(i)$--$(ii)$ is achieved upon the direct calculus of the Poisson brackets
among the separating polynomials $S(u)$, $S(v)$ and the momenta-generating functions~$p(u)$ and~$p(v)$. In order to calculate the corresponding Poisson brackets one needs to find the
 Poisson brackets among the intermediate coordinates $x$, $y$, $f_{\a}$, $g_{\b}$, $h_{\g}$, $\a,\b,\g \in \{1,2,3\}$, in the closed form.
 In order to calculate these brackets it is necessary to find explicitly the derivatives~\smash{$\frac{\partial x}{\partial Y_{ \a}}$} and~\smash{$\frac{\partial y}{\partial Y_{ \a}}$}, where $Y_{\a}=S_{\a}$ or $Y_{\a}=T_{\a}$ or $Y_{\a}=W_{\a}$.
 This is done using the constraint equations~\eqref{xy} in the same way as in Proposition \ref{Zxynul}.

In order to find the Poisson brackets among the coordinates $f_{\a}$, $g_{\b}$, $h_{\g}$, $\a,\b,\g \in \{1,2,3\}$, we use the fact that for any functions $F\equiv F(S_{\a}, T_{\b}, W_{\g}, x,y)$, $G\equiv G(S_{\a}, T_{\b}, W_{\g}, x, y)$ the following equality holds:
\begin{align*}
\{ F, G \}_i={}&\{ F, G \}'_i+
\frac{\partial F}{ \partial x}\sum\limits_{\a,\b=1}^3\biggl( \frac{\partial x}{\partial S_{ \a}}
\frac{\partial G}{\partial S_{ \b}}\{S_{\a},S_{\b}\}_i+
\frac{\partial x}{\partial T_{ \a}} \frac{\partial G}{\partial S_{
\b}}\{T_{\a},S_{\b}\}_i\\
&
+ \frac{\partial x}{\partial W_{ \a}}
\frac{\partial G}{\partial S_{ \b}}\{W_{\a},S_{\b}\}_i +\frac{\partial x}{\partial S_{ \a}}
\frac{\partial G}{\partial T_{ \b}}\{S_{\a},T_{\b}\}_i+ \frac{\partial x}{\partial T_{ \a}} \frac{\partial G}{\partial
T_{ \b}}\{T_{\a},T_{\b}\}_i
\\
& +\frac{\partial x}{\partial W_{ \a}}
\frac{\partial G}{\partial T_{ \b}}\{W_{\a},T_{\b}\}_i
+ \frac{\partial x}{\partial S_{ \a}}
\frac{\partial G}{\partial W_{ \b}}\{S_{\a},W_{\b}\}_i+\frac{\partial x}{\partial T_{ \a}} \frac{\partial G}{\partial
W_{ \b}}\{T_{\a},W_{\b}\}_i\\
&+ \frac{\partial x}{\partial W_{ \a}}
\frac{\partial G}{\partial W_{ \b}}\{W_{\a},W_{\b}\}_i \biggr)+
\frac{\partial F}{ \partial y}
\sum\limits_{\a,\b=1}^3 \biggl( \frac{\partial y}{\partial S_{ \a}}
\frac{\partial G}{\partial S_{ \b}}\{S_{\a},S_{\b}\}_i+
\frac{\partial y}{\partial T_{ \a}} \frac{\partial G}{\partial S_{
\b}}\{T_{\a},S_{\b}\}_i\\
& +
 \frac{\partial y}{\partial W_{ \a}}
\frac{\partial G}{\partial S_{ \b}}\{W_{\a},S_{\b}\}_i+ \frac{\partial y}{\partial S_{ \a}}
\frac{\partial G}{\partial T_{ \b}}\{S_{\a},T_{\b}\}_i
+
 \frac{\partial y}{\partial T_{ \a}} \frac{\partial G}{\partial
T_{ \b}}\{T_{\a},T_{\b}\}_i\\
& +
\frac{\partial y}{\partial W_{ \a}}
\frac{\partial G}{\partial T_{ \b}}\{W_{\a},T_{\b}\}
+ \frac{\partial y}{\partial S_{ \a}}
\frac{\partial G}{\partial W_{ \b}}\{S_{\a},W_{\b}\}_i+
 \frac{\partial y}{\partial T_{ \a}} \frac{\partial G}{\partial
W_{ \b}}\{T_{\a},W_{\b}\}_i\\
&+ \frac{\partial y}{\partial W_{ \a}}
\frac{\partial G}{\partial W_{ \b}}\{W_{\a},W_{\b}\}_i \biggr)-
 \frac{\partial G}{ \partial x} \sum\limits_{\a,\b=1}^3 \biggl( \frac{\partial x}{\partial S_{ \a}}
\frac{\partial F}{\partial S_{ \b}}\{S_{\a},S_{\b}\}_i
+\frac{\partial x}{\partial T_{ \a}} \frac{\partial F}{\partial S_{
\b}}\{T_{\a},S_{\b}\}_i\\
&+ \frac{\partial x}{\partial W_{ \a}}
\frac{\partial F}{\partial S_{ \b}}\{W_{\a},S_{\b}\}_i+ \frac{\partial x}{\partial S_{ \a}}
\frac{\partial F}{\partial T_{ \b}}\{S_{\a},T_{\b}\}_i+ \frac{\partial x}{\partial T_{ \a}} \frac{\partial F}{\partial T_{ \b}}\{T_{\a},T_{\b}\}_i
\\
&+
 \frac{\partial x}{\partial W_{ \a}}
\frac{\partial F}{\partial T_{ \b}}\{W_{\a},T_{\b}\}_i+ \frac{\partial x}{\partial S_{ \a}}
\frac{\partial F}{\partial W_{ \b}}\{S_{\a},W_{\b}\}_i \\
&+ \frac{\partial x}{\partial T_{ \a}} \frac{\partial F}{\partial
W_{ \b}}\{T_{\a},W_{\b}\}_i
+\frac{\partial x}{\partial W_{ \a}}
\frac{\partial F}{\partial W_{ \b}}\{W_{\a},W_{\b}\}_i \biggr)
\\
& -
 \frac{\partial G}{ \partial y} \sum\limits_{\a,\b=1}^3\biggl( \frac{\partial y}{\partial S_{ \a}}
\frac{\partial F}{\partial S_{ \b}}\{S_{\a},S_{\b}\}_i+
\frac{\partial y}{\partial T_{ \a}} \frac{\partial F}{\partial S_{
\b}}\{T_{\a},S_{\b}\}_i+\frac{\partial y}{\partial W_{ \a}}
\frac{\partial F}{\partial S_{ \b}}\{W_{\a},S_{\b}\}_i
\\
&+ \frac{\partial y}{\partial S_{ \a}}
\frac{\partial F}{\partial T_{ \b}}\{S_{\a},T_{\b}\}_i+ \frac{\partial y}{\partial T_{ \a}} \frac{\partial F}{\partial
T_{ \b}}\{T_{\a},T_{\b}\}_i+\frac{\partial y}{\partial W_{ \a}}
\frac{\partial F}{\partial T_{ \b}}\{W_{\a},T_{\b}\}_i\\
&+ \frac{\partial y}{\partial S_{ \a}}
\frac{\partial F}{\partial W_{ \b}}\{S_{\a},W_{\b}\}_i+ \frac{\partial y}{\partial T_{ \a}} \frac{\partial F}{\partial
W_{ \b}}\{T_{\a},W_{\b}\}_i + \frac{\partial y}{\partial W_{ \a}}
\frac{\partial F}{\partial W_{ \b}}\{W_{\a},W_{\b}\}_i \biggr)\\
&+ \biggl( \frac{\partial F}{ \partial x} \frac{\partial G}{ \partial y}- \frac{\partial G}{ \partial x}\frac{\partial F}{ \partial y} \biggr)\{x,y\}_i, \qquad i\in \{1,2\},
\end{align*}
where $\{ \ , \ \}'_i$ is a parenthesis $\{\ , \ \}_i$ in which $x$ and $y$ are treated as constant non-dynamical parameters
and the brackets $\{x,y\}_i$, in their turn, are calculated as follows:
\begin{align*}
\{ x, y \}_i={}&
\sum\limits_{\a,\b=1}^3\biggl( \frac{\partial x}{\partial S_{ \a}}
\frac{\partial y}{\partial S_{ \b}}\{S_{\a},S_{\b}\}_i+
\frac{\partial x}{\partial T_{ \a}} \frac{\partial y}{\partial S_{
\b}}\{T_{\a},S_{\b}\}_i+ \frac{\partial x}{\partial W_{ \a}}
\frac{\partial y}{\partial S_{ \b}}\{W_{\a},S_{\b}\}_i
 \\
 &+ \frac{\partial x}{\partial S_{ \a}}
\frac{\partial y}{\partial T_{ \b}}\{S_{\a},T_{\b}\}_i+ \frac{\partial x}{\partial T_{ \a}} \frac{\partial y}{\partial
T_{ \b}}\{T_{\a},T_{\b}\}_i +\frac{\partial x}{\partial W_{ \a}}
\frac{\partial y}{\partial T_{ \b}}\{W_{\a},T_{\b}\}_i
\\
&+ \frac{\partial x}{\partial S_{ \a}}
\frac{\partial y}{\partial W_{ \b}}\{S_{\a},W_{\b}\}_i+ \frac{\partial x}{\partial T_{ \a}} \frac{\partial y}{\partial
W_{ \b}}\{T_{\a},W_{\b}\}_i +\frac{\partial x}{\partial W_{ \a}}
\frac{\partial y}{\partial W_{ \b}}\{W_{\a},W_{\b}\}_i \biggr).
\end{align*}

Having calculated the Poisson brackets $\{\ ,\ \}_i$, $i\in \{1,2\}$, among the functions $x$, $y$, $f_{\a}$, $g_{\b}$, $h_{\g}$, $\a,\b,\g \in \{1,2,3\}$, in terms of the coordinates $x$, $y$, $S_{\a}$, $T_{\b}$, $W_{\g}$, $\a,\b,\g \in \{1,2,3\}$, we apply the formulae \eqref{STW} and recalculate the right-hand sides of these brackets in terms of the functions $x$, $y$, $f_{\a}$, $g_{\b}$, $h_{\g}$, $\a,\b,\g \in \{1,2,3\}$.

After that, taking into account the explicit form of the functions $S(u)$, $p(u)$ in terms of the functions $x$, $y$, $f_{\a}$, $g_{\b}$, $h_{\g}$, $\a,\b,\g \in \{1,2,3\}$, the constraints \eqref{mconstr}, the definition \eqref{calfa} of the constants $c_{\a}$ and the direct calculations, we come to the following equalities:
\begin{subequations}\label{basic}
\begin{gather}
\{ S(u), S(v)\}_i= 0 \mod \mathcal{J}_{S(u), S(v)},
\\
\{ S(u), p(v)\}_i= 0 \mod \mathcal{J}_{S(u), S(v)}, \qquad u\neq v,
\\
\{ p(u), p(v)\}_i= 0 \mod \mathcal{J}_{S(u), S(v)},
\end{gather}
\end{subequations}
as well as the following equalities:{\samepage
\begin{gather}
\lim_{v\rightarrow u} \{ S(u), p(v)\}_1= \partial_{u} S(u) \mod \mathcal{J}_{S(u)}, \nonumber\\
 \lim_{v\rightarrow u} \{ S(u), p(v)\}_2= - u \partial_{u} S(u) \mod \mathcal{J}_{S(u)}.\label{basic'}
\end{gather}}%

\noindent
Here $\mathcal{J}_{S(u), S(v)}$, $\mathcal{J}_{S(u)}$ are ideals in the space of functions generated by $S(u)$, $S(v)$ and $S(u)$, respectively.
As it follows from the results of \cite{SkrDub2018} the equalities \eqref{qcanon}, \eqref{qcanon1}, \eqref{qcanon2} follow from the equalities \eqref{basic}--\eqref{basic'}.

 This proves the items $(i)$ and $(ii)$ of the theorem.

In order to prove item $(iii)$ of the theorem, we use that for any function $F\equiv F(S_{\a}, T_{\b}, W_{\g},\allowbreak x,y)$ the following equality holds:
\begin{align*}
Z(F)={}& \sum\limits_{\a=1}^3\left( A_{\a}\frac{\partial F}{ \partial S_{\a}}+
B_{\a}\frac{\partial F}{ \partial T_{\a}}+ D_{\a}\frac{\partial F}{ \partial W_{\a}}\right) + \frac{\partial F}{ \partial x}\sum\limits_{\a=1}^3\left( A_{\a}\frac{\partial x}{ \partial S_{\a}}+
B_{\a}\frac{\partial x}{ \partial T_{\a}}+ D_{\a}\frac{\partial x}{ \partial W_{\a}}\right)
\\
&+
\frac{\partial F}{ \partial y} \sum\limits_{\a=1}^3\left( A_{\a}\frac{\partial y}{ \partial S_{\a}}+
B_{\a}\frac{\partial y}{ \partial T_{\a}}+ D_{\a}\frac{\partial y}{ \partial W_{\a}}\right).
\end{align*}
Then, using the derivatives $\frac{\partial x}{\partial Y_{ \a}}$ and $\frac{\partial y}{\partial Y_{ \a}}$, where $Y_{\a}=S_{\a}$ or $Y_{\a}=T_{\a}$ or $Y_{\a}=W_{\a}$, calculated as it is explained above, the explicit form of the components $A_{\a}$, $B_{\a}$, $D_{\a}$ given by the formulas~\eqref{A-B}, the explicit form of the coordinates and momenta of separation given in the text of the theorem, the constraint equations \eqref{xy}, the definition \eqref{calfa} of the constants $c_{\a}$, after tedious calculations we obtain the equalities
$
Z(S(u))=0$, $ Z(p(v))=0$, $ \forall u,v \in \mathbb{C}$.
From these equalities, the equalities \eqref{qpinv} immediately follow.
This proves item $(iii)$ of the theorem.
\end{proof}

\begin{Remark}

 The difference in sign of in the formulae \eqref{qcanon2} with respect to that in \eqref{qcanon} is not crucial and amounts only to the change of sign of spectral parameter of the corresponding Poisson pencil.
\end{Remark}
\subsection{The equations of separation}
In this subsection, we will find equations of separation satisfied by the constructed coordinates~$q_i$ and momenta $p_i$, $i\in \{1,2,3\}$.
The following theorem holds true.
\begin{Theorem}\label{sepeq}
The coordinates $q_i$ as the roots of the polynomial $S(u)$ given by \eqref{Su}--\eqref{Si}, the momenta $p_i$ defined by the formulae
\eqref{pi}--\eqref{pu} and integrals $H$, $K$, $L$, $C_1$, $C_2$, $C_3$ defined by the formulae \eqref{casimir}, \eqref{h}, \eqref{k}, \eqref{l}
satisfy the curve of separation $\mathcal{K}$ of genus five
\begin{gather}
(q_i+j_1)(q_i+j_2)(q_i+j_3) p_i^4+\big(q_i^3 C_3+q_i^2 C_2+ q_i H+K \big) p_i^2\nonumber \\
\qquad{}+ \frac{1}{4}(q_i C_1+L)^2=0,\qquad i\in \{1,2,3\}.\label{eqsepi}
\end{gather}
\end{Theorem}
\begin{proof}
In order to prove the theorem, it is necessary to express the integrals $H$, $K$, $L$, $C_1$, $C_2$, $C_3$ in terms of the intermediate $f-g-h$ coordinate system. %
 Using the explicit form of the integrals in terms of the initial coordinate functions $S_{\a}$, $T_{\a}$, $W_{\a}$, $\a\in \{1,2,3\}$, the formulae \eqref{STW} and the constraints \eqref{mconstr} we, in particular, obtain
\begin{subequations}\label{intinf-g}
\begin{gather}
C_1= -\frac{1}{c_1^2 c_2^2 c_3^2 }\left( \frac{h_1 f_2+ 2 y h_2 h_1+h_3 h_2}{x^2+(y+j_1)(y+j_2)(y+j_3)}+\frac{h_2 g_2}{2 x^2}\right),\label{intinf-g-1}
\\
C_3=\frac{1}{c_1^2 c_2^2 c_3^2 }\left( \frac{ y h_1^2+h_3 h_1}{x^2+(y+j_1)(y+j_2)(y+j_3)} - \frac{h_2^2}{4 x^2} \right),\label{intinf-g-3}
\\
L= -\frac{1}{c_1^2 c_2^2 c_3^2 }\left( \frac{ 2yf_1 h_2+f_3h_2+f_2 f_1}{x^2+(y+j_1)(y+j_2)(y+j_3)}+\frac{g_2 f_2}{2 x^2}\right).
\end{gather}
\end{subequations}
 The explicit expressions for $C_2$, $H$ and $K$ have similar (quadratic in $f_{\a}$, $g_{\a}$, $h_{\a}$ and rational in~$x$ and~$y$ forms). We will not write these expressions explicitly here due to their long form.

Taking into account the explicit formulae
\eqref{intinf-g} and similar formulae for $C_2$, $H$, $K$, substituting them and the formula \eqref{pu} into the right-hand sides of the equalities \eqref{eqsepi}, after long and tedious calculations, we obtain that
\begin{gather*}
(u+j_1)(u+j_2)(u+j_3) p(u)^4+\big(u^3 C_3+u^2 C_2+ u H+K \big) p(u)^2\\
\qquad {}+ \frac{1}{4}(u C_1+L)^2=0 \mod \mathcal{J}_{S(u)}.
\end{gather*}
Here $\mathcal{J}_{S(u)}$ is an ideal generated by separating polynomial $S(u)$ given by the formulae \eqref{Su} and~\eqref{Si}.
\end{proof}

\begin{Remark}
 It is possible to show that the curve $\mathcal{K}$ is equivalent to a spectral curve of a four by four Lax matrix of the extended Clebsch and Manakov models.
\end{Remark}

\begin{Remark}
 Observe, that Theorems \ref{momenti} and \ref{sepeq} assure that the constructed variable separation is a bi-Hamiltonian one and the corresponding vector fields $Z_i$, $i\in \{1,2,3\}$, satisfying~\eqref{symZP}--\eqref{symZQ} do exist.
Nevertheless, the vector fields $Z_i$ are very complicated and are of no practical use. That is why we will not present them here leaving their calculation as an exercise for the interested reader.
\end{Remark}

\begin{Remark}
Note that from the explicit form of the equations of separation together with the constraints $Z(C_1)=0$, $Z(L)=0$ and from the fact that the coordinates $q_i$ and momenta $p_i$, $i\in \{1,2,3\}$, are $Z$-invariants also follows that
 the roots of $S(u)=u^3 Z(C_3)+ u^2 Z(C_2)+u Z(H)+Z(K)$ are separated coordinates. An additional demonstration of this fact is given in the next subsection.
\end{Remark}

\subsection[The vector field Z in the separated coordinates]{The vector field $\boldsymbol{Z}$ in the separated coordinates}
In this subsection, we will explicitly calculate the vector field $Z$ in terms of the coordinates of separation.
Resolving two of the equations \eqref{eqsepi} with respect to $H$ and $K$, we obtain their following form:
\begin{subequations}\label{h-k}
\begin{gather}
H = -\big(q_1^2+q_2q_1+q_2^2\big)C_3-(q_1+q_2)C_2-(q_1+j_1)(q_1+j_2)(q_1+j_3) \frac{p_1^2-p_2^2}{q_1-q_2}\nonumber\\
\phantom{H=}{}+
\frac{1}{4(q_1-q_2)}\left(\frac{(q_2C_1+L)^2}{p_2^2}-\frac{(q_1C_1+L)^2}{p_1^2}\right),
\\
K =(q_1+q_2)q_1q_2C_3+ q_2q_1C_2+(q_1+j_1)(q_1+j_2)(q_1+j_3)\frac{q_2 p_1^2- q_1 p_2^2}{q_1-q_2}
 \nonumber\\
\phantom{K=}{}+
\frac{1}{4(q_1-q_2)}\left(\frac{q_2(q_1C_1+L)^2}{p_1^2}- \frac{q_1(q_2C_1+L)^2}{p_2^2}\right).
\end{gather}
\end{subequations}

Substituting this into the third equation \eqref{eqsepi}, we obtain the following equation:
\begin{gather}
(q_2-q_3) (q_1-q_3) p_3^2 C_2+(q_2-q_3) (q_1-q_3) (q_1+q_2+q_3) p_3^2 C_3+ (q_1+j_3) (q_1+j_2) (q_1+j_1)\nonumber\\
\qquad\times \frac{ (q_2-q_3)}{(q_1 -q_2)} p_3^2 p_1^2
-(q_2+j_3) (q_2+j_2) (q_2+j_1)\frac{(q_1-q_3) }{(q_1-q_2)} p_3^2 p_2^2+(q_3+j_1) (q_3+j_2)\nonumber\\
\qquad\times (q_3+j_3) p_3^4+ \frac{1}{4} (q_3 C_1+L)^2-\frac{1}{4} (q_2 C_1+L)^2
\frac{(q_1-q_3) p_3^2}{(q_1-q_2) p_2^2}\nonumber\\
\phantom{\qquad\times}{}+ \frac{1}{4}(q_1 C_1+L)^2 \frac{(q_2-q_3) p_3^2}{(q_1-q_2) p_1^2}=0.\label{eqsep3'}
\end{gather}

Taking into account that
$
Z(q_1)=Z(q_2)=Z(q_3)=Z(p_1)=Z(p_2)=Z(p_3)= Z(C_1)=0
$
and adding the normalization condition $Z(C_3)=1$, we will look for the vector field $Z$ in the following form
$
Z=Z_3+c_2(q_1,q_2,q_3) Z_2$.
Acting on the equation \eqref{eqsep3'} by $Z$ and finding from the resulting equation $Z(L)$, we obtain
\begin{equation*}
Z(L)= \frac{2 (q_2-q_3)(q_1-q_3)(q_1-q_2) p_1^2p_2^2p_3^2((q_1+q_2+q_3)+c_2(q_1,q_2,q_3))}{ (q_1-q_2)(q_3C_1+L)p_1^2 p_2^2 +(q_3-q_1)(q_2C_1+L) p_1^2 p_3^2+(q_2-q_3)(q_1C_1+L)p_3^2p_2^2 }.
\end{equation*}

From this expression, it immediately follows that the condition $Z(L)=0$ yields the equality~${c_2(q_1,q_2,q_3)=-(q_1+q_2+q_3)}$, i.e.,
we obtain the following simple expression for the vector field $Z$:
\[
Z=Z_3-(q_1+q_2+q_3) Z_2.
\]

Acting by the defined as above vector field $Z$ on the Casimir functions $C_3$, $C_2$ and the integrals~$H$,~$K$ given by \eqref{h-k}, we obtain
$
Z(C_3)=1$, $ Z(C_2)=-(q_1+q_2+q_3)$, $ Z(H)= (q_1q_2+q_1q_3+q_2q_3)$, $ Z(K)= -q_1q_2q_3$,
which again demonstrates that for the given equations of separation $S(u)=u^3 Z(C_3)+ u^2 Z(C_2)+u Z(H)+ Z(K)$ is a polynomial-separator with the roots $q_1$, $q_2$, $q_3$.

\subsection{The Abel-type equations}
The most important for the integration of the equations of motion is possibility to represent these equations in the Abel-type form.
As it follows from the general theory exposed in the Section~\ref{general}, more exactly, from the formula \eqref{abel0}, the following differential equations for the coordinates $q_i$ hold true:
\begin{subequations}\label{AbelCleExt}
\begin{gather}
\sum\limits_{i=1}^3 \frac{2 q_i p_i^3 }{ 4 (q_i+j_1) (q_i+j_2)(q_i+j_3) p^4_i-(q_i C_1+L)^2}\frac{ \partial q_i}{\partial t_j} = \d_{1j},
\\
\sum\limits_{i=1}^3 \frac{2 p_i^3 }{ 4 (q_i+j_1) (q_i+j_2)(q_i+j_3) p^4_i-(q_i C_1+L)^2} \frac{ \partial q_i}{\partial t_j} = \d_{2j},
\\
\sum\limits_{i=1}^3 \frac{ (q_i C_1+L) p_i }{ 4 (q_i+j_1) (q_i+j_2)(q_i+j_3) p^4_i-(q_i C_1+L)^2}\frac{ \partial q_i}{\partial t_j} = \d_{3j},
\end{gather}
\end{subequations}
 where $j\in \{1,2,3\}$, $t_1$, $t_2$, $t_3$ are the parameters along the flows of the integrals $H$, $K$ and $L$ correspondingly
\[
\frac{\partial q_i}{\partial t_1}=\{H,q_i\}_1, \qquad \frac{\partial q_i}{\partial t_2}=\{K,q_i\}_1, \qquad \frac{\partial q_i}{\partial t_3}=\{L,q_i\}_1,
\]
and we have used the equation of separation \eqref{eqsepi} in order to simplify the form of the differentials on the curve $\mathcal{K}$ entering into the equations \eqref{AbelCleExt}. Using the equations \eqref{AbelCleExt}, we easily obtain the Abel-type quadratures written in the differential form as follows:
\begin{subequations}\label{AbelCleExt'}
\begin{gather}
\sum\limits_{i=1}^3 \frac{2 q_i p_i^3 \mathrm{d}q_i}{ 4 (q_i+j_1) (q_i+j_2)(q_i+j_3) p^4_i-(q_i C_1+L)^2} = \mathrm{d} t_1,
\\
\sum\limits_{i=1}^3 \frac{2 p_i^3 \mathrm{d} q_i}{ 4 (q_i+j_1) (q_i+j_2)(q_i+j_3) p^4_i-(q_i C_1+L)^2} = \mathrm{d} t_2,
\\
\sum\limits_{i=1}^3 \frac{ (q_i C_1+L) p_i \mathrm{d} q_i}{ 4 (q_i+j_1) (q_i+j_2)(q_i+j_3) p^4_i-(q_i C_1+L)^2} = \mathrm{d} t_3.
\end{gather}
\end{subequations}

\begin{Remark}
 Note that the Abel-type equations for the extended Manakov model has the same (modulo the overall sign) form as the equations \eqref{AbelCleExt} due to the bi-Hamiltonian equivalence of this model with the Clebsch model. The difference is that in the Manakov case the time flows~$t_1$,~$t_2$,~$t_3$ correspond to the brackets $\{\ ,\ \}_2$ and the integrals $C_2$, $H$ and $C_1$, but the Hamiltonian flows themselves (modulo the overall sign) are the same.
\end{Remark}

\section{Conclusion and discussion}\label{sec6}
In the present paper,
using the method of vector field $Z$ \cite{MFP}, we have constructed {\it symmetric}, {\it non-St\"ackel} variable separation for three-dimensional extension of the Clebsch and Manakov models, for which all curves of separation
 are the same and have genus five.
 We have explicitly constructed the coordinates and momenta of separation 
 and Abel-type quadratures in the considered examples of symmetric SoV for the extended Clebsch and Manakov models.

 We would like also to remark, that our recent results \cite{SkrJGP2025} on separation of variables for the Clebsch model can be re-obtained by the restriction of the construction of this paper onto the six dimensional subspace of Clebsch/Manakov models. By other words, the results of the present paper give also a bi-Hamiltonian explanation to the new variable separation for the Clebsch model constructed in \cite{SkrJGP2025}.

Finally, we would like to outline the following interesting open problems:
\begin{enumerate}\itemsep=0pt
\item[(1)] To find explicit solution of the Abel--Jacobi inversion problem for the Abel-type equations~\eqref{AbelCleExt'} in terms of theta-functions of Prym variety.
 \item[(2)] To obtain the generalization of the results of the present paper onto the higher-dimensional
 extensions of the Clebsch and Manakov models.
 \end{enumerate}

\subsection*{Acknowledgements} The author is grateful to Franco Magri for explaining of the bi-Hamiltonian approach to separation of variables, for the interest to the work and for useful discussions. The author is also grateful to the anonymous referees for their corrections, permitting him to improve the text of the article. The research described in this paper was made possible in part by Isaac Newton Institute and London Mathematical Society Solidarity grant. The author expresses his gratitude to the grant-givers.


\pdfbookmark[1]{References}{ref}
\LastPageEnding


\begin{thebibliography}{99}
\footnotesize\itemsep=0pt

\bibitem{Agost}
Agostinelli C., Sopra l'integrazione per separazione di variabili
 dell'equazione dinamica di {H}amilton--{J}acobi, \textit{Memorie Acc. Scienze
 Torino} \textbf{15} (1937), 3--54.

\bibitem{Benenti}
Benenti S., Intrinsic characterization of the variable separation in the
 {H}amilton--{J}acobi equation,
 \href{https://doi.org/10.1063/1.532226}{\textit{J.~Math. Phys.}} \textbf{38}
 (1997), 6578--6602.

\bibitem{Chanu1}
Benenti S., Chanu C., Rastelli G., Variable separation for natural
 {H}amiltonians with scalar and vector potentials on {R}iemannian manifolds,
 \href{https://doi.org/10.1063/1.1340868}{\textit{J.~Math. Phys.}} \textbf{42}
 (2001), 2065--2091.

\bibitem{Chanu2}
Benenti S., Chanu C., Rastelli G., Remarks on the connection of between the
 additive separation of the {H}amilton--{J}acobi equation and the
 multiplicative separation of the Schroedinger equation.~{II}. First integrals
 and symmetry operators,
 \href{https://doi.org/10.1063/1.1506181}{\textit{J.~Math. Phys.}} \textbf{43}
 (2002), 5183--5222.

\bibitem{BlaMar}
Blaszak M., Marciniak K., Algebraic curves as a source of separable
 multi-{H}amiltonian systems,
 \href{https://doi.org/10.46298/ocnmp.12861}{\textit{Open Commun. Nonlinear
 Math. Phys.}} \textbf{2} (2024), 12861, 27~pages,
 \href{http://arxiv.org/abs/2401.04673}{arXiv:2401.04673}.

\bibitem{Clebsch}
Clebsch A., Ueber die {B}ewegung eines {K}\"orpers in einer {F}l\"ussigkeit,
 \href{https://doi.org/10.1007/BF01443985}{\textit{Math. Ann.}} \textbf{3}
 (1870), 238--262.

\bibitem{SkrDub2018}
Dubrovin B., Skrypnyk T., Separation of variables for linear {L}ax algebras and
 classical {$r$}-matrices,
 \href{https://doi.org/10.1063/1.5031769}{\textit{J.~Math. Phys.}} \textbf{59}
 (2018), 091405, 39~pages.

\bibitem{MFP}
Falqui G., Magri F., Pedroni M., Bihamiltonian geometry and separation of
 variables for {T}oda lattices,
 \href{https://doi.org/10.2991/jnmp.2001.8.s.21}{\textit{J.~Nonlinear Math.
 Phys.}} \textbf{8} (2001), suppl.~1, 118--127,
 \href{http://arxiv.org/abs/nlin.SI/0002008}{arXiv:nlin.SI/0002008}.

\bibitem{FP}
Falqui G., Pedroni M., On a {P}oisson reduction for {G}elfand--{Z}akharevich
 manifolds, \href{https://doi.org/10.1016/S0034-4877(02)80068-4}{\textit{Rep.
 Math. Phys.}} \textbf{50} (2002), 395--407,
 \href{http://arxiv.org/abs/nlin.SI/0204050}{arXiv:nlin.SI/0204050}.

\bibitem{FP2}
Falqui G., Pedroni M., Separation of variables for bi-{H}amiltonian systems,
 \href{https://doi.org/10.1023/A:1024080315471}{\textit{Math. Phys. Anal.
 Geom.}} \textbf{6} (2003), 139--179,
 \href{http://arxiv.org/abs/nlin.SI/0204029}{arXiv:nlin.SI/0204029}.

\bibitem{FedMagSkr2}
Fedorov Yu., Magri F., Skrypnyk T., A new approach to separation of variables
 for the {C}lebsch integrable system. Part~{II}: {I}nversion of the
 {A}bel--{P}rym map, \href{http://arxiv.org/abs/2102.03599}{arXiv:2102.03599}.

\bibitem{GZ}
Gelfand I.M., Zakharevich I., Webs, {L}enard schemes, and the local geometry of
 bi-{H}amiltonian {T}oda and {L}ax structures,
 \href{https://doi.org/10.1007/PL00001387}{\textit{Selecta Math.~(N.S.)}}
 \textbf{6} (2000), 131--183,
 \href{http://arxiv.org/abs/math.DG/9903080}{arXiv:math.DG/9903080}.

\bibitem{Levi}
Levi-Civita T., Sulla integrazione della equazione di {H}amilton--{J}acobi per
 separazione di variabili,
 \href{https://doi.org/10.1007/BF01445149}{\textit{Math. Ann.}} \textbf{59}
 (1904), 383--397.

\bibitem{MagSkr}
Magri F., Skrypnyk T., The {C}lebsch system,
 \href{http://arxiv.org/abs/1512.04872}{arXiv:1512.04872}.

\bibitem{SklSep}
Sklyanin E.K., Separation of variables: {N}ew trends,
 \href{https://doi.org/10.1143/PTPS.118.35}{\textit{Progr. Theoret. Phys.
 Suppl.}} \textbf{118} (1995), 35--60,
 \href{http://arxiv.org/abs/solv-int/9504001}{arXiv:solv-int/9504001}.

\bibitem{SkrJGP2019}
Skrypnyk T., ``{S}ymmetric'' separation of variables for the {C}lebsch system,
 \href{https://doi.org/10.1016/j.geomphys.2018.09.014}{\textit{J.~Geom.
 Phys.}} \textbf{135} (2019), 204--218.

\bibitem{SkrJMP2021}
Skrypnyk T., On a class of~{$gl(n)\otimes gl(n)$}-valued classical
 {$r$}-matrices and separation of variables,
 \href{https://doi.org/10.1063/5.0041967}{\textit{J.~Math. Phys.}} \textbf{62}
 (2021), 063508, 27~pages.

\bibitem{SkrSto}
Skrypnyk T., Symmetric and asymmetric separation of variables for an integrable
 case of the complex {K}irchhoff's problem,
 \href{https://doi.org/10.1016/j.geomphys.2021.104418}{\textit{J.~Geom.
 Phys.}} \textbf{172} (2022), 104418, 26~pages.

\bibitem{SkrBD}
Skrypnyk T., Symmetric and asymmetric separated variables.~{II}. {A} second
 integrable case of the complex {K}irchhoff model,
 \href{https://doi.org/10.1063/5.0152882}{\textit{J.~Math. Phys.}} \textbf{64}
 (2023), 093503, 21~pages.

\bibitem{SkrAsSoVExCle}
Skrypnyk T., Asymmetric separation of variables for the extended {C}lebsch and
 {M}anakov models,
 \href{https://doi.org/10.1016/j.geomphys.2023.105078}{\textit{J.~Geom.
 Phys.}} \textbf{197} (2024), 105078, 18~pages.

\bibitem{SkrJGP2025}
Skrypnyk T., Separation of variables for the {C}lebsch model:~{$so(4)$}
 spectral/separation curves,
 \href{https://doi.org/10.1016/j.geomphys.2025.105453}{\textit{J.~Geom.
 Phys.}} \textbf{212} (2025), 105453, 11~pages.

\bibitem{Sta1}
St\"ackel P., Ueber die {B}ewegung eines {P}unktes in einer {$n$}-fachen
 {M}annigfaltigkeit, \href{https://doi.org/10.1007/BF01447379}{\textit{Math.
 Ann.}} \textbf{42} (1893), 537--563.

\bibitem{Sta2}
St\"ackel P., Ueber quadratische Integrale der Differentialgleichungen der
 Dynamik, \href{https://doi.org/10.1007/BF02580501}{\textit{Ann. Mat. Pura
 Appl.}} \textbf{26} (1897), 55--60.

\end{thebibliography}
\end{document}